\documentclass[12pt,a4paper]{article}

\usepackage[T1]{fontenc}
\usepackage[utf8]{inputenc}

\usepackage{amsmath, amssymb, amsthm}

\usepackage[margin=1in]{geometry}
\usepackage{setspace}
\usepackage{hyperref}
\usepackage{natbib}

\usepackage{graphicx}
\graphicspath{{Fig/}}
\usepackage{algorithm}
\usepackage{algpseudocode}

\usepackage{enumitem}

\theoremstyle{plain}
\newtheorem{theorem}{Theorem}

\theoremstyle{definition}
\newtheorem{definition}{Definition}

\theoremstyle{remark}

\newcommand{\Grad}{\nabla\!\!\!\!\nabla}
\newcommand{\rd}{\mathrm{d}}

\title{Newton's Algorithm as a Gradient Flow:\\
       A Geometric Framework for Recursive Mixture Estimation}

\author{Bernardo Flores\\[4pt]
        \small Department of Statistics \& Data Sciences,
        University of Texas at Austin\\
        \small \href{mailto:bernardofl@utexas.edu}{bernardofl@utexas.edu}}

\date{}

\begin{document}
\maketitle

\begin{abstract}
Bayesian nonparametric mixture models provide a flexible framework for data analysis but are often hindered by the computational expense of traditional inference methods like MCMC. A fast, recursive algorithm proposed by Newton (2002) offers a practical alternative, yet its formal connection to Bayesian inference and its theoretical properties remain only partially understood. This paper reveals a new geometric interpretation of this class of predictive recursions. We demonstrate that Newton's recursion is a discrete-time approximation of a gradient flow on the space of probability measures governed by the Fisher-Rao geometry, providing the first rigorous dynamical characterisation of this family of estimators. This geometric perspective provides a principled theoretical foundation for studying these recursions: it clarifies their convergence behaviour, situates them within the variational Bayes literature, and yields a systematic basis for generalisation by modifying the underlying geometry and discretisation. In contrast to approaches that construct gradient flows from a prescribed variational objective, this work proceeds in the reverse direction: beginning from an existing recursive estimator and uncovering the variational problem it implicitly solves, it opens a pathway for the systematic analysis and extension of a broad class of sequential Bayesian estimators.
\end{abstract}


\bigskip
\section{Introduction}
Bayesian nonparametric mixture models offer a powerful and flexible approach for density estimation and clustering in heterogeneous datasets. Despite their theoretical appeal, their practical application is often hampered by significant computational challenges. Standard inference methods, such as Markov Chain Monte Carlo (MCMC), can exhibit poor mixing on the multimodal posterior landscapes characteristic of these models, requiring many iterations to approach stationarity.

As an efficient alternative, \cite{newton2002nonparametric} introduced a recursive algorithm for estimating the mixing measure in a density estimation problem. The algorithm sequentially updates the estimate with each new data point, is computationally fast, easy to implement, and has been shown to be strongly consistent. Its Bayesian properties, however, are not fully characterised: while the update weights can be matched to those of a Pólya urn scheme of a Dirichlet process \citep{ferguson_bayesian_1973}, the finite-sample estimates are not guaranteed to approximate the true posterior.

The conceptual roots of Newton's algorithm lie in the quasi-Bayes sequential procedure of \citet{smith1978quasi}, who proposed a recursive update rule for unsupervised classification in finite mixtures and established its frequentist consistency via stochastic approximation techniques. Newton revisited this idea in the nonparametric setting in \cite{newton1998predictive}, showing that the recursive update step corresponds exactly to the predictive distribution of a Dirichlet process mixture, thereby placing the earlier quasi-Bayes heuristic on a formal nonparametric Bayesian footing. Further work with collaborators, including an extension to missing-data problems \citep{newtonzhang1999}, culminated in the 2002 Sankhyā paper \citep{newton2002nonparametric}, which synthesised these developments and proved strong consistency of the mixing measure estimate under general conditions on the kernel family. Subsequent follow-up work established a firmer theoretical foundation for the algorithm: \citet{martinghosh2008} placed the recursion within the stochastic approximation framework; \citet{tokdarmartin2009} proved almost-sure consistency in the weak topology under weaker regularity conditions; and \citet{martintokdar2009, martintokdar2011} characterised the rate of convergence and extended the approach to semiparametric inference via a predictive recursion marginal likelihood.

These contributions established Newton's algorithm as one of the first concrete instances of predictive recursion: a class of inference procedures in which the posterior is built recursively through one-step-ahead predictive distributions rather than computed from a joint probability model \citep{recursiveHahn,  berti2021models}. In such procedures, the predictive distribution, not the posterior, is taken as the primitive inferential object, anticipating the modern literature on predictive Bayesian inference, \emph{c.f.} \citep{cappellowalker2025, fong2023martingale, kerneldirichlet, berti2025probabilistic}. The question of how faithfully Newton's predictive recursion approximates the exact Bayesian posterior was taken up by \citet{fortini2020}, who showed that the sequence of predictive distributions induced it satisfies the same asymptotic guarantees as the Bayesian posterior without converging to it in finite samples. This precise characterisation motivates the present work, understanding the geometric mechanism driving the recursion is a necessary step toward controlling and improving its Bayesian fidelity.

In this paper, we reveal a deeper geometric structure underlying this classic algorithm. We demonstrate that Newton's recursion is a discrete-time approximation of a time-dependent gradient flow on the space of probability measures, specifically with respect to the Fisher-Rao geometry. This reinterpretation provides a principled perspective from which to understand, analyse, and extend the original algorithm as a variational optimisation problem. By framing the estimation problem as the minimisation of an energy functional on a statistical manifold, we derive a more general family of algorithms by modifying the underlying geometry, the energy functional, or the discretisation scheme.

The remainder of this paper is organised as follows. Section~\ref{sec:mixture} reviews Bayesian mixture models and the properties of Newton's algorithm. Section~\ref{sec:primer} introduces the theory of gradient flows on metric spaces. Section~\ref{sec:gradflow} presents our main theoretical result connecting the algorithm to a Fisher-Rao gradient flow. Section~\ref{sec:newtonsmith} introduces the Newton-Smith framework of extended algorithms. Section~\ref{sec:experiments} provides numerical experiments, followed by a concluding discussion in Section~\ref{sec:conclusion}.

\section{Mixture models and the recursive estimator}\label{sec:mixture}

We consider Bayesian mixture models of the form
\begin{equation}\label{eq:mixture}
    x_{1:n} \sim \int k(x,\theta)\,\mu(\rd \theta); \quad \mu \sim \mathbb{P},
\end{equation}
where $\mathbb{P}$ is a prior over probability measures --- for instance, a Dirichlet Process \citep{ferguson_bayesian_1973}, a nonparametric prior whose finite-dimensional marginals are Dirichlet-distributed and which generates almost-surely discrete random measures --- with mean measure $P$. These models take a simple kernel $k(x,\theta)$ and propagate it across the data space by fitting a collection of local components. Under a discrete mixing measure $\mu$, \eqref{eq:mixture} implies inference on the clustering of units $1,\dots, n$, achieved by introducing latent indicator variables that link each unit $i$ to an atom of $\mu$.

\cite{newton2002nonparametric} introduced an efficient recursive estimator for the mixing measure $\mu$. Starting with an initial guess $\mu_0$ and a squared-summable sequence $(\alpha_n)$ on $(0,1)$, the estimate is updated by
\begin{equation}\label{eq:newton}
    \mu_{n+1}(\rd \theta) = (1-\alpha_n)\mu_n(\rd \theta) + \alpha_n\frac{\displaystyle k(x_{n+1},\theta)\,\mu_n(\rd \theta)}{\int k(x_{n+1},\theta)\,\mu_n(\rd\theta)}.
\end{equation}

Algorithm \eqref{eq:newton} yields a strongly consistent point estimator: as $n\to\infty$, $\mu_n$ converges in probability to the true mixing measure $\mu^*$. Its Bayesian properties, however, are not fully characterised. The $\alpha_n$ weights can be matched to the constants in the Pólya urn for a Dirichlet process or any species sampling model. However, \cite{fortini2020} showed that this does not guarantee that the finite-sample induced estimate $\mu_n$ is close to the posterior (mean) under the DP mixture model \eqref{eq:mixture}.

We can rewrite the recursion \eqref{eq:newton} as
\begin{equation}\label{eq:newton2}
    \mu_{n+1}(\rd \theta) = \mu_n(\rd \theta) + \alpha_n\left[\frac{\displaystyle k(x_{n+1},\theta)\,\mu_n(\rd \theta)}{\int_{\Theta} k(x_{n+1},\theta)\,\mu_n(\rd\theta)}-\mu_n(\rd \theta)\right].
\end{equation}

Equation~\eqref{eq:newton2} starts to look like gradient ascent optimisation for the probability measure $\mu$, with the gradient direction corresponding to the term inside the square brackets. The next section reviews gradient ascents on the space of probability measures, studied through their time-continuous limit as gradient flows.

\section{A Primer on Gradient Flows and Statistical Geometries}\label{sec:primer}

Gradient flows provide a natural framework for describing the continuous-time evolution of systems that decrease an energy functional.

In an Euclidean space \( \Theta\subseteq \mathbb{R}^d \), if \( \mathcal{F}:\Theta \to \mathbb{R} \) is a smooth, convex function, its gradient flow is the path \(\theta(t)\) defined by the ordinary differential equation (ODE):
\begin{equation}\label{eq:gd}
\theta'(t) = -\nabla \mathcal{F}(\theta(t)), \quad \theta(0) = \theta_0.
\end{equation}
Here, \( -\nabla \mathcal{F}(\theta(t)) \) is the direction of steepest descent, and $\theta'(t)=\frac{\rd}{\rd t}\theta(t)$.

A common approach to approximating this continuous path is with the implicit Euler method. This discrete-time rule can be recast as a minimisation problem:
\[
\theta_{n+1} = \underset{\theta \in \Theta}{\mathrm{arg\,min}} \left[ \frac{|\theta - \theta_n|^2}{2\tau} + \mathcal{F}(\theta) \right].
\]

Setting the gradient of the objective to zero recovers the discrete-time approximation of \eqref{eq:gd}. This formulation is known as the minimising movement scheme: it finds a new point $\theta_{n+1}$ minimising $\mathcal{F}(\theta)$ subject to a squared-distance penalty from $\theta_n$.

This scheme generalises naturally. Let $\theta_n^\tau$ denote a discrete time approximation of $\theta(t)$, using time steps of size $\tau$. To define a gradient flow in a generic metric space $( \Theta, d)$, one replaces the Euclidean distance with the metric $d$:
\begin{equation}\label{eq:jko}
\theta_{n+1}^\tau \in \underset{\theta \in \Theta}{\mathrm{arg\,min}} \left[ \frac{d(\theta, \theta_n^\tau)^2}{2\tau} + \mathcal{F}(\theta) \right].
\end{equation}
As the time step \( \tau \to 0 \), this discrete path converges to a continuous trajectory defined as the gradient flow of \( \mathcal{F} \) in the metric space \( (\Theta, d) \) \citep{jordan1998variational}. For a comprehensive treatment, we refer the reader to \citet{ambrosio2005gradient}.

We now consider the special case of the space of probability distributions over $\Theta$, denoted $\mathcal{P}(\Theta)$. Two metrics used to induce the appropriate structures to ensure well-definedness are the 2-Wasserstein and the Fisher-Rao distances.  The direction of steepest descent within the chosen metric is described by a differential operator $\Grad_d$ called affine connection.

The general flow is described by the partial differential equation (PDE) $\partial_t \mu_t =- \Grad_d \mathcal{F}(\mu_t)$. While typically solutions of it pertain only to Lebesgue densities; a classic generalisation in PDE theory called weak solutions can incorporate discrete probability measures. Formally, a weak solution is a measure-valued path $(\mu_t)$ for which, for every smooth test function $\varphi \in C_c^\infty(\Theta)$, the map $t \mapsto \int \varphi \, 
d\mu_t$ is absolutely continuous and satisfies the equation after integration against $\varphi$ \citep{evans2010pde}. To characterise these equations explicitly, we introduce the first variation of a functional $\mathcal{F}$.

\begin{definition}
    Let $\chi$ be a signed measure with $\int \rd \chi = 0$. The first variation of a functional $\mathcal{F}:\mathcal{P}(\mathcal{X})\to\mathbb{R}$ at $\mu$ is the unique function $\frac{\delta \mathcal{F}}{\delta \mu}: \mathcal{X} \to \mathbb{R}$ (defined up to an additive constant) satisfying:
    $$
    \frac{\rd}{\rd \epsilon} \mathcal{F}(\mu + \epsilon \chi) \bigg|_{\epsilon=0} = \int \frac{\delta \mathcal{F}}{\delta \mu}(x) \, \chi(\rd x).
    $$
\end{definition}

We write both $\mu(\rd \theta)$ for the measure and $\mu(\theta)$ for its associated density.

\subsection{Geometries on the Space of Probability Measures}

A particularly amenable family of minimising movement schemes in the space of distributions are those given by the free energy functionals  $\mathcal{F}:\left[\mathcal{P}(\Theta), d\right]\to \mathbb{R}$, which take the form:
\begin{equation}\label{eq:freeE}
    \mathcal{F}(\mu)=\underbrace{\int V(\theta)\,\mu(\rd \theta)}_{\text{Potential}}+ \underbrace{\int \log\mu(\theta)\,\mu(\rd\theta)}_\text{Internal energy}.
\end{equation}
For Bayesian inference, minimising the Kullback-Leibler (KL) divergence $\mathcal{F}(\mu) = \text{KL}(\mu \| \pi_n)$ corresponds to a potential $V(\theta) = -\log \pi_n(\theta)$.

\subsubsection{The Wasserstein Geometry: Moving Mass}
Our first candidate for the metric $d$ is the 2-Wasserstein distance, $W_2$ \citep{villani2009optimal}. In this geometry, the tangent space consists of vector fields that transport mass. The affine connection, denoted $\Grad_W$, can be shown to correspond to the continuity equation \citep{jordan1998variational}:
\begin{equation}\label{eq:gradw}
\partial_t \mu_t = \nabla \cdot \left[\mu_t \nabla \frac{\delta \mathcal{F}}{\delta \mu}\right].
\end{equation}
This equation represents the conservation of probability: as the distribution evolves, no mass is created or destroyed, only transported.

\subsubsection{The Fisher-Rao Geometry: Reshaping Density}
This geometry provides a distinct notion of distance. The Fisher-Rao metric \citep{rao1945information, fishermanifold} measures the distinguishability between distributions. The affine connection $\Grad_{FR}$ is a scalar field representing the relative growth rate \citep{lieroHellinger}:
$$
\partial_t \mu_t = -\mu_t\left[ \frac{\delta \mathcal{F}}{\delta \mu} - \int \frac{\delta \mathcal{F}}{\delta \mu} \,\rd\mu_t \right].
$$
While Wasserstein flows transport mass, Fisher-Rao gradient flows re-weight it.

\subsubsection{The Wasserstein-Fisher-Rao Geometry: A Hybrid Approach}
\label{sec:hk}
The Wasserstein-Fisher-Rao geometry combines these into a single product metric \citep{Liero2018, lieroHellinger, jkosplitting}. The tangent space at a measure $\mu$ decomposes into two orthogonal components: $\mathcal{T}_\mu \mathcal{P}(\mathcal{X}) \cong \mathcal{T}_\mu^{\text{W}} \oplus \mathcal{T}_\mu^{\text{FR}}$. Consequently, the WFR affine connection can be written as the direct sum of the individual gradients:
$$
\Grad_{WFR} \mathcal{F} = \Grad_W \mathcal{F} \oplus \Grad_{FR} \mathcal{F}.
$$
The WFR continuity equation reflects this composite structure and can be shown to be \citep{Liero2018, jkosplitting}:
\begin{equation}\label{eq:contHK}
\partial_t \mu_t + \underbrace{\nabla \cdot \left[\mu_t \Grad_W \mathcal{F} \right]}_{\text{Transport}} = \underbrace{\mu_t \Grad_{FR} \mathcal{F}}_{\text{Reaction}}.
\end{equation}
Here $\Grad_W \mathcal{F} = -\nabla \frac{\delta \mathcal{F}}{\delta \mu}$ and $\Grad_{FR} \mathcal{F} = -\left[\frac{\delta \mathcal{F}}{\delta \mu} - \int \frac{\delta \mathcal{F}}{\delta \mu}\,\rd\mu\right]$ denote the transport velocity and reaction rate, respectively, so that the flow $\partial_t \mu_t = -\Grad_{WFR}\mathcal{F}$ decomposes into a conservative flux and a source term.

\subsection{Implementation of the Gradient Flows}

Our goal is to evolve an initial measure $\mu_0$ towards the target posterior $\pi_n$. Unlike classical MCMC, gradient flows provide a dynamic transformation perspective. The path $(\mu_t)$ is the exact solution to a PDE that minimises the free energy; as $t\to\infty$, $\mu_t$ converges to $\pi_n$.

There are two broad strategies for numerically implementing these flows, distinguished by how the affine connection $\Grad_d \mathcal{F}$ is treated. Analytic methods treat the governing PDEs as purely deterministic. Probabilistic methods instead recognise these PDEs as Fokker--Planck equations describing the density evolution of specific Markov processes, whose invariant measure is exactly the target posterior $\pi_n$.

A canonical example of the probabilistic viewpoint is the Wasserstein gradient flow with $\mathcal{F}(\mu)=\text{KL}(\mu \| \pi_n)$. Equation \eqref{eq:gradw} then characterises a Langevin diffusion, so running an unadjusted Langevin Monte Carlo algorithm is equivalent to solving an infinite-dimensional optimisation problem \citep{jordan1998variational}. We use probabilistic methods throughout this paper.

\subsubsection{Discrete Approximation and Mean-Field Limits}
We approximate the continuous posterior $\pi_n$ using the empirical measure of a system of $N$ particles, $\hat{\mu}_t = \frac{1}{N} \sum_{i=1}^{N} \delta_{\theta_i(t)}$, or more generally a weighted $\hat{\mu}_t=\sum w_t \delta_{\theta_i(t)}$.
While this introduces a discretisation error, McKean-Vlasov theory provides a rigorous justification: as $N \to \infty$, the interactions average out (propagation of chaos \citep{chaintron2021}), and $\hat{\mu}_t$ converges exactly to the solution of the mean-field PDE.

\begin{enumerate}[label=(\alph*)]
    \item \textbf{Wasserstein (Moving Atoms):} Since weights are fixed, the flow must represent regions of high posterior density by physically transporting particles. To represent a heavy mixture component, the flow \textit{clumps} atoms together.
    \item \textbf{Fisher-Rao (Reweighting):} Particle locations are fixed. The flow represents posterior mass by increasing the weights of atoms. This creates a \textit{coverage constraint}: to successfully represent a mode, the system requires particles to be already present in that region.
\end{enumerate}

\subsubsection{The Hybrid Algorithm: Wasserstein-Fisher-Rao}
We employ a splitting scheme that alternates between these dynamics \citep{crucinio2025sequentialmontecarloapproximations}:

\begin{enumerate}[label=(\alph*)]
    \item Fisher-Rao Step (Birth-Death): Update weights of current particles $\{\theta_i(t)\}$ based on the Fisher-Rao gradient:
    $$
    \tilde{w}_i \propto w_i(t) \exp\left[ \tau \Grad_{FR} \mathcal{F}(\theta_i(t)) \right].
    $$
    Then resample $N$ particles $\{\hat{\theta}_i\}$ with replacement. This effectively concentrates mass on high-fitness regions.

    \item Wasserstein Step (McKean-Vlasov): Evolve locations of resampled particles $\{\hat{\theta}_i\}$ using the Euler-Maruyama discretisation:
    $$
    \theta_i(t+\tau) = \hat{\theta}_i - \tau \left[ \nabla V(\hat{\theta}_i)\right] + \sqrt{2\tau} \xi_i,
    $$
    where $\xi_i \sim \mathcal{N}(0, I)$. This step transports particles via potential and interaction forces.
\end{enumerate}

\section{A Gradient Flow Perspective of Newton's Algorithm}\label{sec:gradflow}

We are now ready to recognise Newton's recursion~\eqref{eq:newton} as an instance of a (discretised) gradient flow. We restate recursion \eqref{eq:newton} using unified notation:
$$
    \mu_{n+1}(\rd \theta) = \mu_n(\rd \theta) + \alpha_n\left[\frac{\displaystyle k(x_{n+1},\theta)\,\mu_n(\rd \theta)}{\int_{\Theta} k(x_{n+1},\theta)\,\mu_n(\rd\theta)}-\mu_n(\rd \theta)\right].
$$
This expression has the form of a forward Euler discretisation. By replacing the discrete measure $\mu_n$ with a continuous density $\mu_t(\theta)$ and $\alpha_n$ with a time increment, the limit is given by the PDE:
\begin{equation}\label{eq:contFR}
    \partial_t \mu_{t}(\theta) = \mu_t(\theta)\left[\frac{\displaystyle k(x_{t},\theta)}{\int_{\Theta} k(x_{t},\vartheta)\mu_t(\vartheta)\,\rd\vartheta}-1\right].
\end{equation}
This purely multiplicative equation is characteristic of a Fisher-Rao gradient flow.

\begin{theorem}
Newton's recursion \eqref{eq:newton} is a forward Euler discretisation of the Fisher-Rao gradient flow of the time-dependent functional
\begin{align}\label{eq:newtop}
\mathcal{F}_t(\mu_t\mid x_t) &= -\log \int k(x_t, \theta)\,\mu_t(\rd\theta).
\end{align}
\end{theorem}

\begin{proof}
First, we calculate the first variation $\delta\mathcal{F}$. Let $\chi$ be a signed measure such that $\int\rd \chi=0$. Then,
    \begin{align*}
    \lim_{\epsilon\to 0} \frac{\mathcal{F}_t(\mu+\epsilon \chi\mid x_t) -\mathcal{F}_t(\mu\mid x_t) }{\epsilon}
    &= -\lim_{\epsilon\to 0}\frac{\log\left[\int k(x_t,\theta)\,(\mu+\epsilon\chi)(\rd\theta)\right] - \log\int k(x_t,\theta)\,\mu(\rd\theta)}{\epsilon}\\
    &= - \frac{\int k(x_t,\theta)\,\chi(\rd\theta)}{\int k(x_t,\vartheta)\,\mu(\rd\vartheta)} = -\int \frac{ k(x_t,\theta)}{\int k(x_t,\vartheta)\,\mu(\rd\vartheta)}\,\chi(\rd\theta).
    \end{align*}
Hence, $\frac{\delta \mathcal{F}_t}{\delta \mu}(\theta) = -\frac{k(x_t,\theta)}{\int k(x_t,\vartheta)\,\mu(\rd\vartheta)}$.

The Fisher-Rao gradient flow is given by:
\begin{align*}
\partial_t \mu_t &= \mu_t \Grad_{FR} \mathcal{F}_t(\mu_t) = -\mu_t\left[\frac{\delta\mathcal{F}_t}{\delta \mu}-\int\frac{\delta\mathcal{F}_t}{\delta \mu}\,\rd\mu_t\right]\\
&=-\mu_t\left[\frac{\delta\mathcal{F}_t}{\delta \mu}+\int \frac{k(x_t,\vartheta)}{\int k(x_t,\varphi)\,\rd\mu_t(\varphi)} \,\rd\mu_t(\vartheta)\right] = -\mu_t\left[\frac{\delta\mathcal{F}_t}{\delta \mu}+1\right].
\end{align*}
This matches \eqref{eq:contFR}, identifying Newton's algorithm as the forward Euler discretisation of this flow.
\end{proof}

This theorem shows that \cite{newton2002nonparametric} is a forward Euler discretisation of the Fisher--Rao gradient flow \eqref{eq:contFR}.

To situate this method within a Bayesian context, we turn to the variational formulation of Bayes' rule. The framework of Generalised Variational Inference \citep{knoblauch2022} casts Bayesian inference as an optimisation problem over probability measures. The key tool is the Donsker--Varadhan variational formula, which states that for any measurable function $\Phi$:
\begin{equation}\label{eq:donsker}
    \log \mathbb{E}_P \left[ e^{\Phi(\theta)} \right] = \sup_{q \in \mathcal{P}(\Theta)} \left[ \mathbb{E}_q \{\Phi(\theta)\} - \text{KL}(q \| P) \right].
\end{equation}
Setting $\Phi(\theta) = \sum \log p(x_i | \theta)$ shows that the posterior is the unique maximiser of the right-hand side, or equivalently the minimiser of the free energy functional:
\begin{equation}\label{eq:bayesvar}
q_n^* = \underset{q \in \mathcal{P}(\Theta)}{\mathrm{arg\,min}} \left[ \mathbb{E}_{q} \left\{ -\sum_{i=1}^n \log(p(x_i \mid \theta)) \right\} + \text{KL}(q \| P) \right].
\end{equation}

This addresses the concern raised by \cite{fortini2020}: to make \eqref{eq:newton} approximate the true posterior under a prior $\mathbb{P}$, one augments the functional $\mathcal{F}$ with the discrepancy term from \eqref{eq:bayesvar}.

However, directly adding the KL term to fix Newton's recursion is not an option since, in the case of infinite mixtures, standard divergences -like the KL are often ill-defined between random measure priors, such as the prior $\mathbb{P}$ in \eqref{eq:mixture}. We therefore adopt a generalised Bayesian framework, replacing the KL with the Wasserstein-over-Wasserstein distance $\mathcal{W}_2(\mathbb{P},\mathbb{Q})$ \citep{nguyenborrowing}. We define the discrepancy as:
$$
\mathcal{W}(q\,\|\, \mathbb{P}) :=\mathcal{W}_2(\delta_q,\mathbb{P})^2,
$$
which lifts the deterministic measure $q$ to a Dirac mass in $\mathcal{P}(\mathcal{P}(\mathcal{X}))$. Since $q$ is deterministic, the only possible coupling is the independent copula $\delta_q\otimes \mathbb{P}$, simplifying the distance to:
$$
\mathcal{W}(q\,\|\, \mathbb{P}) =\int W_2^2(q, p)\,\mathbb{P}(\rd p).
$$
This distance was studied in \cite{hugomarta} as a discrepancy between species sampling models \citep{Pitman1996SomeDO}, which includes popular priors such as the Dirichlet and Pitman--Yor processes.

The standard 2-Wasserstein distance $W_2$ provides a robust geometry, but its dual potentials are piecewise linear when the measures are discrete. This lack of strict convexity poses a theoretical challenge for defining unique gradient flows. To remedy this, we replace the inner transport cost with the entropically regularised Wasserstein (Sinkhorn) discrepancy, $W_{2,\varepsilon}^2$. The discrepancy becomes:
$$
\mathcal{W}_\varepsilon(q\,\|\, \mathbb{P}) =\int W_{2,\varepsilon}^2(q, p)\,\mathbb{P}(\rd p).
$$
Since the prior $\mathbb{P}$ is typically a random measure (e.g., a Dirichlet Process), this integral is analytically intractable but can be efficiently estimated via Monte Carlo sampling. We approximate the discrepancy by drawing $M$ i.i.d. realisations $p_k \sim \mathbb{P}$ and computing the empirical average:
$$
\hat{\mathcal{W}}_\varepsilon(q\,\|\, \mathbb{P}) \approx \frac{1}{M} \sum_{k=1}^M W_{2,\varepsilon}^2(q, p_k).
$$
As shown by \cite{cuturi2013sinkhorn} and \cite{feydy2019interpolating}, the entropic term ensures that the corresponding Kantorovich potentials are smooth ($C^\infty$) and strictly convex, guaranteeing the existence and uniqueness of the gradient flow even for discrete particle approximations.

Substituting the Kullback--Leibler divergence with the Sinkhorn-regularised Wasserstein metric places us within the framework of Generalised Variational Inference (GVI) \citep{knoblauch2022, bissiri2016general}. Unlike standard Bayesian inference, which is uniquely characterised by minimisation of the relative entropy, GVI defines the posterior as the solution to a variational problem with an arbitrary loss and regularizer. The measure recovered by our algorithm is therefore strictly a generalised posterior.

This departure is theoretically necessary in our context. As noted by \cite{nguyenborrowing}, standard Bayesian updating is ill-defined for infinite mixing measures, whereas the geometry-aware Sinkhorn distance provides robust regularisation that remains well-posed in infinite-dimensional metric spaces. This leads to a general definition of this class of algorithms.

\begin{definition}[Newton-Smith Recursion]
Let $\mathcal{D}(\cdot, \mathbb{P})$ be a generic discrepancy measure with respect to a prior $\mathbb{P}$, and let $\lambda \geq 0$ be a regularisation parameter. The Newton-Smith recursion is defined as a temporal discretisation of the gradient flow of the regularised functional:
\begin{equation}\label{eq:NQ_functional}
    \mathcal{J}_t(\mu) = \mathcal{F}_t(\mu \mid x_t) + \lambda \mathcal{D}(\mu, \mathbb{P}).
\end{equation}
The dynamics are governed by the chosen geometry $\mathcal{G}$:
\begin{equation}\label{eq:NQ_flow}
    \partial_t \mu_t = -\Grad_\mathcal{G} \mathcal{J}_t(\mu_t).
\end{equation}
\end{definition}

\section{Extending the Algorithm: The Newton-Smith Framework}\label{sec:newtonsmith}

The basic Newton recursion is computationally efficient but geometrically restricted: it evolves weights while keeping support points fixed. It also minimises a functional with no explicit prior mechanism, which is why its connection to Bayesian inference is tenuous. This section extends the algorithm by introducing  prior regularisation and lifting the static support restriction via a hybrid geometry. We term this broader class of updates the Newton-Smith framework.

\subsection{Regularisation via Optimal Transport}

Standard recursive estimation can be ill-posed without proper regularisation. We focus on the case where the functional includes a discrepancy term $\mathcal{D} = \mathcal{W}_\varepsilon$ regularised by $\lambda=1$. This corresponds to a Generalised Variational Inference (GVI) objective where the prior is enforced via the Sinkhorn-over-Wasserstein distance.

To implement this, we require the first variation of the regularisation term. Standard results on the Sinkhorn distance \citep{feydy2019interpolating} identify this variation as the expectation of the optimal dual potentials. Restricting evolution to the weights (a pure Fisher-Rao flow regularised by the Wasserstein distance), the driving force on the $i$-th particle becomes:
$$
\Grad_{FR} \mathcal{J}_t(\theta_i) = -\left[ V_t(\theta_i) - \bar{V}_t \right],
$$
where $V_t(\theta_i) = \frac{\delta\mathcal{F}_t}{\delta \mu}(\theta_i) + \lambda \int f_{\mu_t \to p}(\theta_i) \mathbb{P}(\rd p)$ is the total potential and $\bar{V}_t$ is its expectation under $\mu_t$. This reduces the infinite-dimensional PDE to a system of ODEs for the weights, which we simulate in Algorithm \ref{alg:particle_hk}, enabling uncertainty quantification via bootstrapping.

\subsection{The Wasserstein-Fisher-Rao Splitting Scheme}

Our second extension lifts the static support restriction by moving from the Fisher-Rao geometry to the Wasserstein-Fisher-Rao (HK) geometry. As detailed in Section \ref{sec:hk}, the gradient flow in this space decomposes into orthogonal \textit{reaction} (weight update) and \textit{transport} (location update) components, allowing the mixture components to evolve jointly with their weights.

To simulate this flow, we employ a splitting scheme that alternates between these two dynamics at every step. Given the current measure $\mu_n = \sum_{j=1}^N w_j(n) \delta_{\theta_j(n)}$, the update to $\mu_{n+1}$ proceeds in two stages:

\begin{enumerate}[label=(\alph*)]
    \item \textbf{Reaction (Fisher-Rao Step):} We first update the weights based on the fitness of each atom. This is an Euler exponential update driven by the potential $V_{i,n}$. To physically realise the birth-death dynamic essential for mode discovery, we then \textbf{resample} the particles proportional to these new weights. This crucial step duplicates high-fitness particles and discards low-fitness ones.

    \item \textbf{Transport (Wasserstein Step):} We then update the locations of the resampled atoms using the Wasserstein affine connection. This moves particles locally towards the likelihood maxima, allowing the support to adapt to the data stream.
\end{enumerate}

This procedure is summarised in Algorithm \ref{alg:particle_hk}. By combining these steps, the algorithm enjoys the benefits of both approaches: the reweighting step can jump between modes, while the gradient descent step provides local spatial refinement.

\begin{algorithm}[ht]
\caption{Particle Wasserstein-Fisher-Rao Flow }
\label{alg:particle_hk}
\begin{algorithmic}[1]
\State \textbf{Input:} Random Prior $\mathbb{P}$, data stream $(x_t)$.
\State \textbf{Parameters:} Bootstrap replicates $B$, particles $N$, Monte Carlo prior samples $M$, Sinkhorn regularisation $\varepsilon$, flow regularisation $\lambda$.

\For{$b = 1$ to $B$}
    \State \textbf{Initialise:} Draw $\{\theta_i^{(b)}\}_{i=1}^N \sim P$; set weights $w_i^{(b)}(0) = 1/N$.

    \For{$k = 0$ to $K-1$}
        \State \textbf{1. Likelihood Force:}
        \State Compute normalised likelihood variation:
        $$
        g_i^{(b)} = -\frac{k(x_{t_k}, \theta_i^{(b)})}{\sum_{j=1}^N w_j^{(b)}(t_k) k(x_{t_k}, \theta_j^{(b)})}
        $$

        \State \textbf{2. Prior Force (Monte Carlo Sinkhorn):}
        \State Sample $M$ realisations from the prior: $\{p_m\}_{m=1}^M \sim \mathbb{P}$.
        \For{$m = 1$ to $M$}
            \State Compute dual Sinkhorn potential $f_{\mu^{(b)} \to p_m}$ between current measure $\mu_{t_k}^{(b)}$ and $p_m$.
        \EndFor
        \State Average potentials to estimate prior variation:
        $$
        h_i^{(b)} = \frac{1}{M} \sum_{m=1}^M f_{\mu^{(b)} \to p_m}(\theta_i^{(b)})
        $$

        \State \textbf{3. Update:}
        \State Total Potential: $V_i^{(b)} = g_i^{(b)} + \lambda h_i^{(b)}$.
        \State Mean Potential: $\bar{V}^{(b)} = \sum_{j=1}^N w_j^{(b)}(t_k) V_j^{(b)}$.

        \State Evolve weights (Fisher-Rao step):
        $$
        \tilde{w}_i^{(b)}(t_{k+1}) = w_i^{(b)}(t_k) \exp\left[ -\Delta t_k \left\{ V_i^{(b)} - \bar{V}^{(b)} \right\} \right]
        $$
        \State Normalise weights: $w^{(b)}(t_{k+1}) \leftarrow \tilde{w}^{(b)}(t_{k+1}) / \sum \tilde{w}^{(b)}(t_{k+1})$.
    \EndFor
    \State Store final measure $\hat{\mu}^{(b)}$.
\EndFor
\end{algorithmic}
\end{algorithm}

Newton-Smith's functional is grounded in the marginalised log-likelihood, operating on a single observation at a time. Substituting the complete log-likelihood recovers a regularised variant of the algorithm in \cite{Yan2023LearningGM}. Neither formulation holds an unambiguous theoretical advantage in general; however, the sequential structure of our approach may make it better suited to large-scale data regimes where batch processing is computationally prohibitive.

\subsection*{Connection to Post-Bayesian Inference}

Before establishing conditions under which the Newton--Smith recursions converge, 
we must address the required properties of the data stream $(x_t)_{t \leq N}$. 
Since $(x_t)$ is indexed by time, convergence of the gradient flow of functional 
\eqref{eq:newtop} depends on the ability to extend this stream to infinity. 
Recall that we embed the discrete sequence $(x_n)_{n=1}^N$ into a continuous-time 
process via step-function interpolation, setting $x_t := x_n$ whenever 
$n - 1 \leq t < n$.

When $N$ is large, finite-horizon effects are negligible and convergence follows 
from standard arguments. The more delicate regime is small $N$, where extension 
of the data stream requires care. The naive extension $x_t = x_N$ for all 
$t \geq N$ is always available, but introduces a spurious stationarity that can 
impair mixing by breaking the exchangeability of the stream. A more principled 
alternative is the Bayesian bootstrap \citep{rubin1981bayesian}, which draws 
$x_{N+k}$, $k > 0$, from a Dirichlet process centred at the empirical 
distribution $\hat{p}(x_1, \dots, x_N)$. Under this construction, we recover the 
method of \citet{fuheng} as a forward Euler discretisation of the 
Wasserstein--Fisher--Rao gradient flow, establishing a direct connection between 
our recursion and that framework.

A second structural observation concerns the prior regularisation term $\mathcal{D}(\mu, \mathbb{P})$, which represents the main computational bottleneck of our approach. By the Bernstein--von Mises phenomenon \citep{vaart1998asymptotics}, when the posterior is consistent, the likelihood term asymptotically dominates the prior as $N \to \infty$, rendering the choice of regularisation asymptotically immaterial. This has a natural interpretation in stochastic optimal control theory, where KL-type regularisation terms shape the trajectories of controlled diffusion processes \citep{wang2020continuous, kappen2005path}. In that setting, it is standard to anneal the regularisation strength via a decreasing schedule $\lambda_t \searrow 0$, yielding the time-dependent objective
\begin{equation}
    \mathcal{J}_t(\mu) \;=\; \mathcal{F}_t(\mu \mid x_t) 
    + \lambda_t \,\mathcal{D}(\mu, \mathbb{P}).
    \label{eq:annealed}
\end{equation}

Such scheduling progressively deactivates the prior as evidence accumulates, mirroring the Bayesian asymptotic regime while preserving stability in finite samples. The rate at which $\lambda_t \to 0$ governs the exploration--exploitation trade-off: too slow a decay retains unnecessary regularisation bias; too fast a decay forfeits the stabilising effect of the prior before sufficient data have been observed. This is directly analogous to temperature scheduling in simulated annealing \citep{geman1984stochastic}, where the schedule $\lambda_t \sim C/\log t$ is the slowest rate that guarantees convergence to the global optimum \citep{holley1988simulated}. Whether such schedules are necessary in our setting, or whether faster rates are permissible under stronger conditions on the likelihood geometry, remains an open question.

Finally, we situate our approach within the broader landscape of post-Bayesian inference \citep{bissiri2016general, knoblauch2022}. To our knowledge, \citet{Yan2023LearningGM} were the first to propose gradient flows for mixture model estimation. Their formulation targets the full log-likelihood $n^{-1}\sum_{i=1}^n \mathcal{F}_n(\theta)$ as the driving functional, in contrast to the sequential marginal likelihood underlying our approach. Supplementing their objective with the regularisation term $\mathcal{D}(\mu, \mathbb{P})$ recovers a gradient flow targeting the Bayesian posterior, bridging both frameworks. 

The question of which formulation to prefer is a fundamental one in post-Bayesian inference \citep{knoblauch2022}. Our sequential, marginal-likelihood approach adapts dynamically and accommodates streaming data naturally. The static full-likelihood alternative may offer greater stability in small-$N$ regimes where the prior exerts a stronger influence. We investigate this trade-off empirically in Section~\ref{sec:experiments}.

\section{Theoretical Properties}\label{sec:theory}

We now study the convergence guarantees of the generalised Newton-Smith recursions. The central challenge is that the Newton-Smith operator \eqref{eq:newtop} is time-dependent, rendering the flow equations non-autonomous. We apply the framework of \cite{FerreiraValenciaGuevara2018} to establish convergence of Algorithm~\ref{alg:particle_hk}. This requires conditions on the first variation of the functional (which depends on the kernel) and on the mean measure of the prior.

The most important condition is $\lambda(t)$-convexity. A function $V(t,\theta)$ is $\lambda$-convex for each $t$ if $V(t, \theta)-\frac{\lambda(t)}{2}\theta^2$ is convex for some locally bounded, positive function $\lambda$. Geometrically, this requires log-likelihoods to have curvature at least as strong as a parabola.

The conditions found in \cite{FerreiraValenciaGuevara2018} applied to our functional \eqref{eq:NQ_functional} $\mathcal{J}_t(\theta)$ are as follows. The functional must be $\lambda(t)$-convex, and the minimal norm subgradient $|\partial_\theta \mathcal{J}_t(\theta)|$ must be locally bounded with $t \mapsto \mathcal{J}_t(\theta)$ locally bounded from below. There must also exist a positive, locally $L_1$ function $\beta(t)$ such that
$$
|\mathcal{J}_t(\theta) - \mathcal{J}_s(\theta)| \leq \int_{s}^{t} \beta(r) \, \rd r \left[1 + \|\theta\|^2\right], \quad 0 \leq s < t,
$$
and a Lebesgue density $\rho(\theta)$ with finite second moment satisfying $\int \mathcal{J}_0(\theta)\, \rho(\theta)\,\rd\theta < +\infty$.

Note that \cite{FerreiraValenciaGuevara2018} apply constraints directly to the score function, the gradient of the log-likelihood. As this gradient serves as the first variation of the KL divergence that governs Langevin dynamics, their technique applies to our framework by substituting the first variation of our own functional. This motivates the subsequent theorem.

\subsection*{Regularity Assumptions}

We make the following standard assumptions on the domain, the prior measure, and the data stream. These follow the regularity theory of Sinkhorn potentials \citep{feydy2019interpolating} and gradient flows \citep{FerreiraValenciaGuevara2018}.

\begin{enumerate}[label=(\alph*)]
\item \textbf{Compact \& Convex Domain.} The parameter space $\Theta \subset \mathbb{R}^d$ is a compact, strictly convex set.

\item \textbf{Bounded Kernel.} The mixture kernel $k(x,\theta)$ belongs to a minimal exponential family with $C^2$ smooth natural parameters.

\item \textbf{Locally Integrable Data Stream.} The sufficient statistics $T(x_t)$ of the data stream possess a locally integrable envelope: for any finite time interval $[0, T]$, $\int_0^T \|T(x_t)\| \, \rd t < \infty$. For discrete data streams or bounded data domains, this condition holds trivially. Note there is no problem with randomness, hence bootstrapping is also covered.
\end{enumerate}
\begin{theorem}
    The Newton-Smith (Wasserstein) gradient flow governing the dynamics of \eqref{eq:NQ_functional} has a unique weak solution if the mixture kernel $k(x, \theta)$ belongs to the exponential family, the parameter space $\Theta$ satisfies the Regularity Assumptions, and the data stream admits a locally integrable envelope.
\end{theorem}
\begin{proof}
    We verify that the first variation of $\mathcal{J}_t(\mu)$ satisfies the regularity and displacement convexity conditions of \cite{FerreiraValenciaGuevara2018}.
    
    The functional $\mathcal{J}_t(\mu) = \mathcal{F}_t(\mu \mid x_t) + \lambda \mathcal{W}_\varepsilon(\mu \parallel \mathbb{P})$ has first variation---the potential $V_t(\theta)$ driving the flow---given by:
    $$
    V_t(\theta) = \frac{\delta \mathcal{J}_t}{\delta \mu}(\theta) = \underbrace{-\frac{k(x_t,\theta)}{\int k(x_t,\vartheta)\,\mu_t(\rd \vartheta)}}_{\text{Likelihood variation } V^{\mathcal{F}}_t(\theta)} + \underbrace{\lambda \int f_{\mu_t \to p}(\theta) \, \mathbb{P}(\rd p)}_{\text{Sinkhorn variation } V^{\mathcal{W}}_t(\theta)}
    $$
    where $f_{\mu_t \to p}(\theta)$ is the optimal dual Sinkhorn potential transporting the current estimate $\mu_t$ to a realisation $p$ of the prior $\mathbb{P}$.

    \paragraph{Step 1: $\lambda$-Convexity of the Potential.}
    We analyse the curvature of $V_t(\theta)$ by bounding the Hessians of its two components.
    
    For the prior variation $V^{\mathcal{W}}_t$, we rely on the properties of entropic optimal transport. As proven in \cite{feydy2019interpolating} (Proposition 1 and Theorem 3), the Sinkhorn potentials $f_{\mu_t \to p}(\theta)$ are $C^\infty$ smooth and strictly convex for any $\varepsilon > 0$. Specifically, their Hessian is strictly positive definite everywhere on the compact domain $\Theta$. Since averaging preserves positive definiteness, the expectation over the prior $\mathbb{P}$ inherits this property. Crucially, the bound is uniform in $\mu_t$: the space $\mathcal{P}(\Theta)$ is compact in the weak topology (since $\Theta$ is compact), and the map $\mu \mapsto \nabla^2 f_{\mu \to p}(\theta)$ is continuous in $\mu$ by stability of entropic optimal transport \citep{feydy2019interpolating}. A continuous, strictly positive-definite-valued function on a compact set attains a uniform lower bound. Thus, there exists a constant $C_P > 0$, independent of $\mu_t$, such that:
    $$
    \nabla^2 V^{\mathcal{W}}_t(\theta) = \lambda \int \nabla^2 f_{\mu_t \to p}(\theta) \, \mathbb{P}(\rd p) \succeq \lambda C_P I.
    $$

    For the likelihood variation $V^{\mathcal{F}}_t$, assume the kernel $k(x,\theta)$ is a member of the exponential family. By the properties of exponential families \cite{brown1986fundamentals}, the density $k(x_t, \theta)$ is $C^2$ smooth with respect to $\theta$. Because $\Theta$ is compact, the Hessian of the likelihood variation is continuous and uniformly bounded. Therefore, its eigenvalues are bounded from below by some time-dependent constant $-K_t \le 0$:
    $$
    \nabla^2 V^{\mathcal{F}}_t(\theta) \succeq -K_t I.
    $$
    
    Summing the two components, the total potential $V_t(\theta)$ satisfies:
    $$
    \nabla^2 V_t(\theta) \succeq (\lambda C_P - K_t) I.
    $$
    This confirms that the potential $V_t(\theta)$ is $\lambda(t)$-convex with modulus $\lambda(t) = \lambda C_P - K_t$. Note that for sufficiently large regularisation $\lambda$, the functional becomes strongly convex.
    
    \paragraph{Step 2: Subgradient Boundedness and Coercivity.}
    Because $\Theta$ is compact and $V_t(\theta)$ is smooth (inheriting $C^2$ smoothness from the exponential family and $C^\infty$ from the Sinkhorn potentials), the minimal norm subgradient $|\partial_\theta V_t(\theta)|$ is strictly bounded on $\Theta$. Furthermore, continuity on a compact set ensures $V_t(\theta)$ is bounded from below for any fixed $t$.

    \paragraph{Step 3: Time Regularity.}
    The time-dependence of the potential enters strictly through the data stream $x_t$ in the likelihood term. Assuming the sufficient statistics $T(x_t)$ possess a locally $L_1$ integrable envelope $\beta(t)$, the variation of the potential is bounded:
    $$
    |V_t(\theta) - V_s(\theta)| \leq \int_{s}^{t} \beta(r) \, \rd r \left[1 + \|\theta\|^2\right], \quad \text{for } 0 \leq s < t.
    $$
    
    With $\lambda(t)$-convexity, bounded gradients, and time-regularity verified, the conditions of \cite{FerreiraValenciaGuevara2018} are satisfied, guaranteeing that the non-autonomous Wasserstein gradient flow \eqref{eq:NQ_flow} admits a unique weak solution.
\end{proof}

While the compactness assumption may appear restrictive, in practice one can always truncate the support of any parameter to a sufficiently large bounded set. A further restriction is that we considered only continuous members of the exponential family, which is standard for mixture modelling since mixing discrete distributions tends to produce misspecified models \citep{miller2014inconsistency}. Finally, we note that the analysis here covers only the Wasserstein gradient flow, without incorporating the Fisher-Rao dynamics. While the full Wasserstein-Fisher-Rao geometry underpins our practical splitting scheme, extending rigorous convergence and consistency guarantees to this hybrid space presents substantial mathematical challenges. The pure Wasserstein space is endowed with a mature differential calculus, allowing us to use established tools from optimal transport to rigorously analyse the spatial evolution of the atoms. In contrast, the Fisher-Rao component introduces non-conservative birth-death processes. Analysing the interacting particle system under these coupled dynamics requires working on a cone over the Wasserstein space, demanding specialised techniques that fall outside the scope of this theoretical exposition. Nevertheless, analysing the pure Wasserstein flow captures the most critical theoretical aspect of our mixture model: how the gradient of the free energy successfully transports particles toward the true posterior. The Fisher-Rao updates, while practically essential for rapidly navigating multimodal landscapes and re-allocating mixture weights, can be viewed theoretically as an auxiliary selection mechanism that accelerates the spatial convergence guaranteed by the underlying Wasserstein diffusion.

\section{Numerical Experiments}\label{sec:experiments}

We evaluate the Newton-Smith flows through two experiments implemented in the accompanying codebase.\footnote{\url{https://github.com/BernardoFL/RecursiveMixtures}} The experiments are designed in sequence: the first establishes which geometry matters for multimodal recovery; the second, having identified the WFR flow as the strongest performer, asks whether two implementation choices that the geometry itself leaves open are worth making.

To isolate the effect of the update geometry, we use a target that is deliberately hard: a seven-component Gaussian mixture arranged in the silhouette of a paw. A dataset of $n = 1000$ observations is drawn i.i.d. from this mixture, and all three flows are initialised from the same particle measure drawn from a Pitman--Yor prior $\mathrm{PY}(d, \theta, G_0)$ \citep{pitmanyor1997} --- a two-parameter extension of the Dirichlet process in which discount $d \in [0,1)$ controls the rate of new cluster formation and concentration $\theta > -d$ governs the total mass --- with $d = 0.2$, $\theta = 10.0$, $G_0 = \mathcal{N}(0, 25 I)$, and $N = 50$ particles. Each flow consumes the $n$ observations sequentially with no bootstrapping, so any difference in the final configurations is attributable solely to the update rule. Neither the Fisher-Rao flow nor the base Newton recursion moves atoms; both are permanently constrained to the initial support. The WFR flow breaks this constraint by coupling a Fisher-Rao reweighting step with a Sinkhorn-regularised Wasserstein transport step ($\varepsilon = 0.05$, $\rho = 0.03$, $\lambda = 0.05$, $25$ inner iterations, prior flow weight $0.1$).

Figure~\ref{fig:flow_comparison} shows the outcome. Both gradient-flow samplers outperform the Newton recursion, which explores the density landscape poorly. The Wasserstein transport term is the decisive factor: by allowing atoms to physically relocate toward regions of high density, the WFR flow recovers the full paw structure that the weight-only methods cannot reach from the initial diffuse support.

\begin{figure}[htbp]
    \centering
    \includegraphics[width=\textwidth]{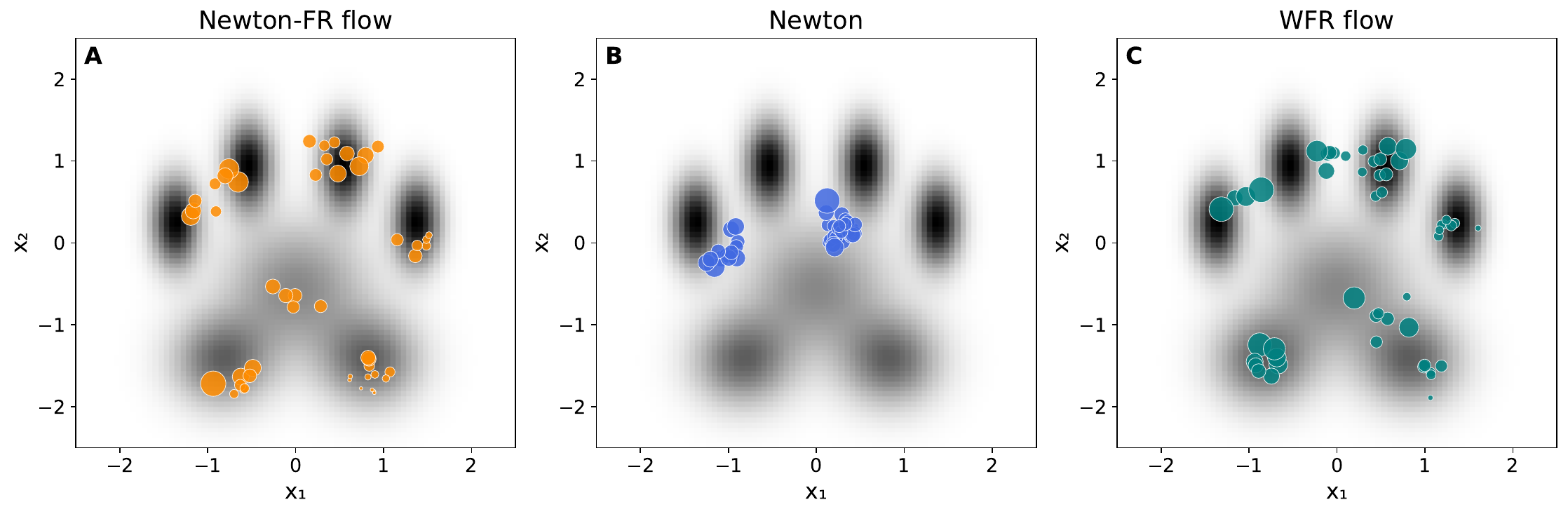}
    \caption{%
        Flow comparison on the seven-component cat-paw Gaussian mixture.
        All three panels share the same initialisation and data.
        Background: true target density (grayscale heatmap, darker $=$
        higher density). Small markers: observed data scatter. Large
        markers: final particle locations with area proportional to weight.
        Left: FR flow. Centre: Newton recursion. Right: WFR flow.
    }
    \label{fig:flow_comparison}
\end{figure}

Having established that the WFR geometry dominates, we ask whether two degrees of freedom that the geometry itself leaves undetermined are consequential in practice. Both comparisons use a four-component Gaussian mixture target, with the WFR flow initialised from $\mathrm{PY}(0.2, 10.0, \mathcal{N}(0, 9))$ at $N = 50$ particles and evaluated across $n \in \{100, 200, \ldots, 1000\}$. For each arm and sample size, a Bayesian bootstrap replicate is formed by drawing Dirichlet weights over the $n$ data indices and constructing a resampled stream of the same length.

The first question is whether extending the flow beyond a single data pass via index continuation improves the result. A truncated arm halts after consuming the $n$ bootstrap observations exactly once; a continuation arm draws further steps via a Bayesian bootstrap for $\lceil 1.5n \rceil$ total steps. We performed 100 repeated simulations and show in Figure~\ref{fig:nxw2} the mean and 90\% CI bands of the Wasserstein distance between the final particle measures and the true distribution for both arms in logarithmic scale.  We also show in Figure~\ref{fig:bootstrap} two examples of these simulations with $n=100, 1000$, which show how adding the bootstrapped samples helps convergence.

\begin{figure}[htbp]
    \centering
    \includegraphics[width=0.7\linewidth]{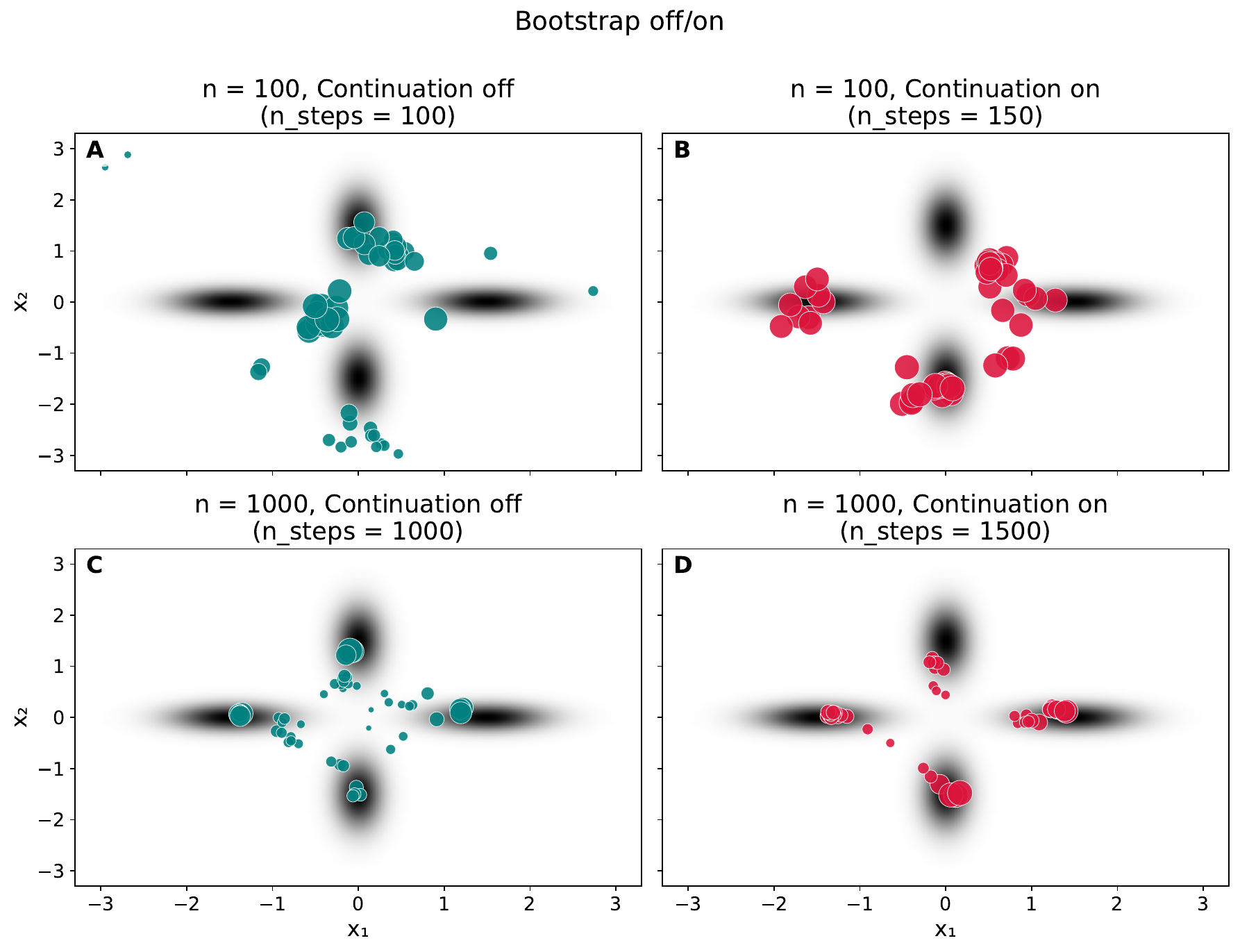}
    \caption{Mean $W_2$ distance between the empirical particle measure and the true distribution as a function of sample size $n$, for truncated (orange) and continuation (blue) bootstrap arms, with 90\% confidence bands (log scale).}
    \label{fig:bootstrap}
\end{figure}

The second question is whether activating the prior regularisation term is worth the extra Sinkhorn solve at each step. Both arms run for $n$ steps with no index continuation; the only difference is whether the Sinkhorn prior functional contributes to the Fisher-Rao gradient or is suppressed, leaving only the likelihood gradient to drive the weights. Atom-level Sinkhorn drift toward the prior remains active in both arms. Figures~\ref{fig:choice} and~\ref{fig:nxw2} tell a consistent story: the final particle configurations are visually similar at both $n = 100$ and $n = 1000$, but the $W_2$ trajectories show that prior regularisation provides a consistent convergence advantage across all sample sizes.

\begin{figure}[htbp]
    \centering
    \includegraphics[width=0.75\textwidth]{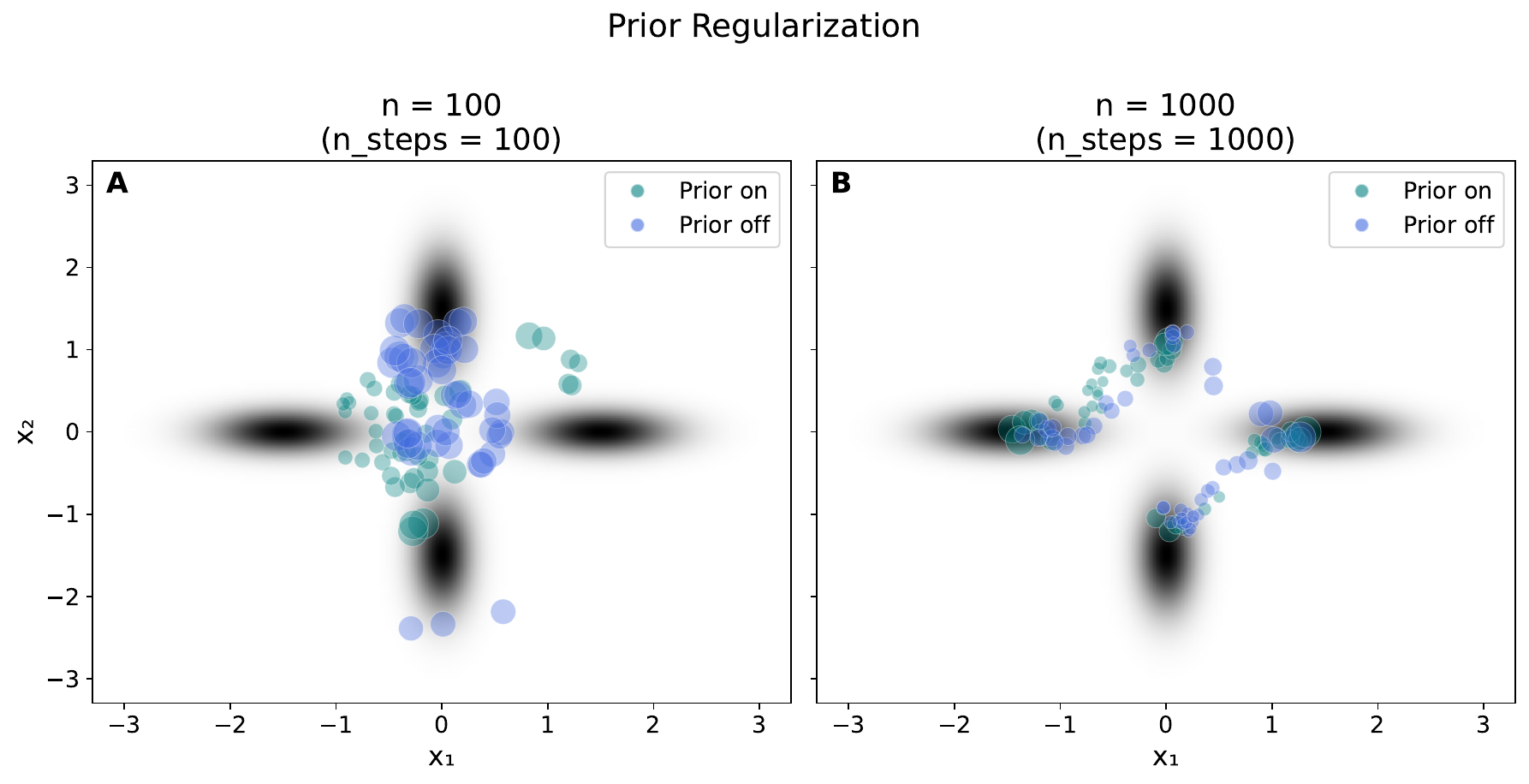}
    \caption{%
        WFR particle configurations with prior regularisation
        enabled (teal) and disabled (royal blue), at sample sizes
        $n \in \{100, 1000\}$, displayed as a $2 \times 2$ grid.
        Both arms use $n$ steps with no index continuation.
        Background: true density heatmap; marker area proportional to weight.
    }
    \label{fig:choice}
\end{figure}

\begin{figure}[ht]
    \centering
    \includegraphics[width=0.8\linewidth]{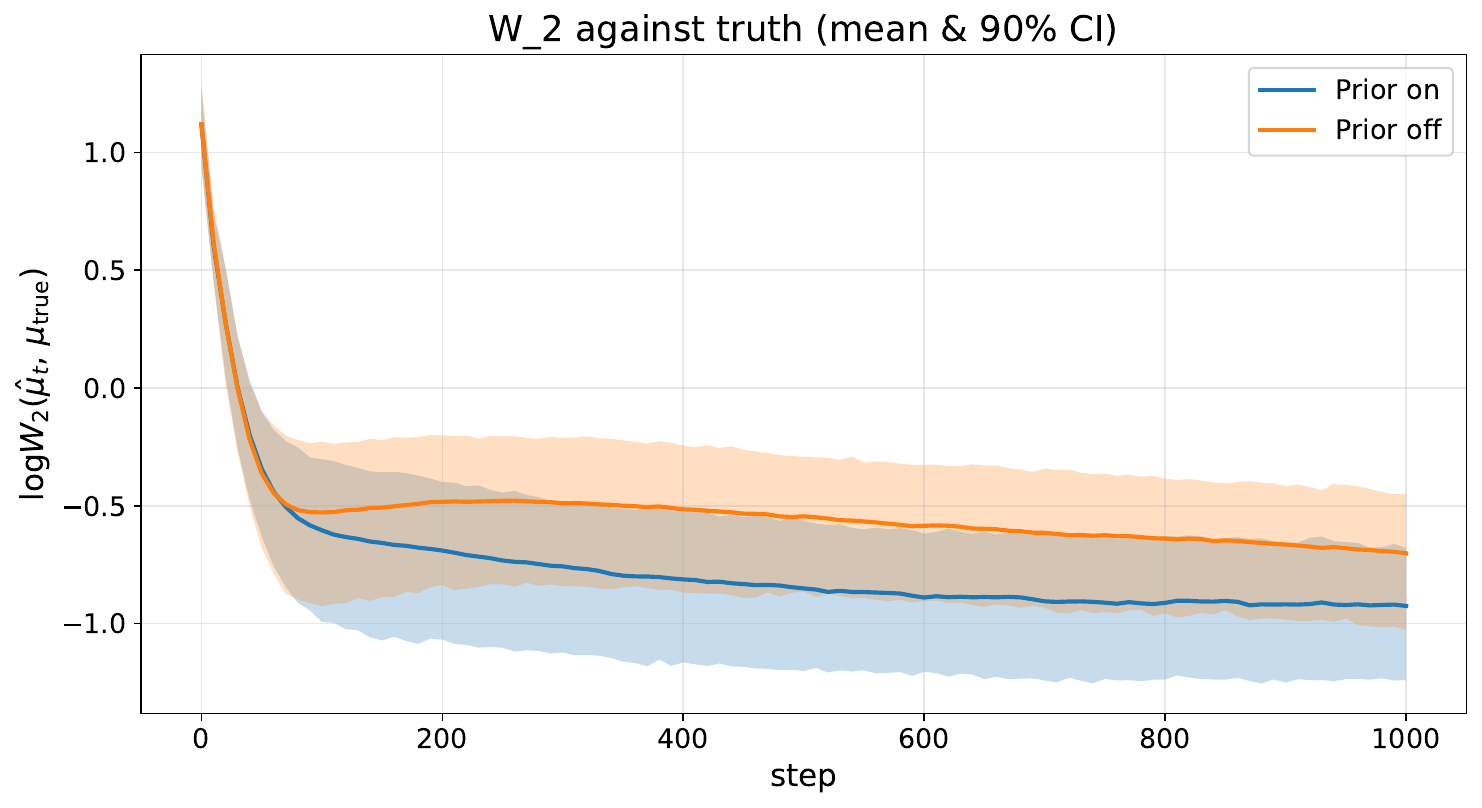}
    \caption{Mean $W_2$ distance between the empirical particle measure and the true distribution as a function of sample size $n$, with 90\% bootstrap confidence bands (log scale). Teal: prior regularisation enabled; royal blue: disabled.}
    \label{fig:nxw2}
\end{figure}

\section{Conclusion and Future Work}\label{sec:conclusion}

In this paper, we introduced an online particle inference framework for Bayesian mixture models grounded in optimal transport. By lifting sequential inference to a gradient flow on the space of probability measures, our approach continuously minimises a variational formulation of Bayes' theorem via a regularised marginal log-likelihood.

Our core contribution generalises the algorithm of \cite{newton2002nonparametric} by deriving its time-continuous limit and constructing a computationally tractable splitting scheme for the Wasserstein-Fisher-Rao gradient flow. The birth-death mechanism tunnels through low-probability regions, while Wasserstein transport provides precise local shape matching. Together, these dynamics mitigate the mode collapse that frequently affects standard MCMC and variational methods.

A current limitation is the use of a Wasserstein-over-Wasserstein discrepancy for nonparametric priors. This requires solving an optimal transport problem at every iteration, incurring significant computational overhead, and is not the most natural discrepancy between random measures. Exploring alternatives, such as the variational $f$-divergences studied by \cite{NIPS2007_72da7fd6} or the mean KL divergences exposed in \cite{ramses}, constitutes an important direction for future work on principled variational inference in infinite-dimensional spaces.

While this work establishes the foundational geometry and algorithmic structure, the framework is highly modular. In a forthcoming companion paper, we extend this methodology to address several practical challenges in applied Bayesian inference. Specifically, we adapt the flow for dependent mixture models, where the Wasserstein step transitions into a covariate-dependent Mean-Field Langevin diffusion. We also introduce repulsive mixture components to encourage cluster separation and develop mini-batching strategies for large datasets.

A further remark concerns recursive estimators of this type more broadly. Such estimators are now prevalent in post-Bayesian inference, arising naturally as conditionally identically distributed (CID) sequences---including martingale posteriors \citep{fong2023martingale} and other urn-like processes \citep{kernelbased}. Recasting these constructions as optimisation problems on the space of probability measures is a natural direction that could yield novel theoretical and practical insights.
\section*{Acknowledgments}
\noindent I wish to express my gratitude to my advisor Peter Müller and to Khai Nguyen for their guidance, valuable discussions, and constructive feedback throughout this research.

\bibliography{references}
\bibliographystyle{abbrvnat}

\end{document}